\documentclass[12pt,reqno]{amsart}

\usepackage{amsaddr}
\usepackage{graphicx} 

\newtheorem{claim}{Claim}[section]
\newtheorem{theorem}[claim]{Theorem}

\newtheorem{definition}[claim]{Definition}

\numberwithin{equation}{section}

\begin{document}
\title[Pseudo orbit expansion for the resonance condition]{Pseudo orbit expansion for the resonance condition on quantum graphs and the resonance asymptotics}
\author{Ji\v{r}\'{i} Lipovsk\'{y}}
\address{Department of Physics, Faculty of Science, University of Hradec Kr\'{a}lov\'{e}, Rokitansk\'{e}ho 62, 500\,03 Hradec Kr\'{a}lov\'{e}, Czechia}
\email{jiri.lipovsky@uhk.cz}

\begin{abstract}
In this note we explain the method how to find the resonance condition on quantum graphs, which is called pseudo orbit expansion. In three examples with standard coupling we show in detail how to obtain the resonance condition. We focus on non-Weyl graphs, i.e. the graphs which have fewer resonances than expected. For these graphs we explain benefits of the method of ``deleting edges'' for simplifying the graph.
\end{abstract}

\maketitle


\section{Introduction}
Resonances are a phenomenon, which occurs often in physics and can be easily understood heuristically. Nevertheless, studying it mathematically rigorously is more difficult. There are two main definitions of resonances~-- \emph{resolvent resonances} (poles of the meromorphic continuation of the resolvent into the non-physical sheet) and \emph{scattering resonances} (poles of the meromorphic continuation of the scattering matrix). Non-compact quantum graph, where halflines are attached to a compact part of the graph (see figure~\ref{fig0}), provides a good background for studying resonances. It has been proven in \cite{EL1} that on quantum graph the above two definitions are almost equivalent; to be precise, the set of resolvent resonances is equal to the union of the set of scattering resonances and the set of the eigenvalues supported only on the internal part of the graph. There is a large bibliography on resonances in quantum graphs; for the quantum chaos community e.g. the papers \cite{KoS3, KSc} might be interesting.

The pseudo orbit expansion is a powerful tool for the trace formula expansion and the secular equation on compact quantum graphs. We refer the reader to \cite{BHJ} and the references therein. The method has been recently adjusted to finding the resonance condition for non-compact quantum graphs \cite{Li1}. 

The resonance asymptotics on non-compact quantum graphs was first studied in \cite{DP}, where it was observed that some graphs do not obey expected Weyl behaviour and a criterion for distinguishing non-Weyl graphs with standard coupling has been obtained. This criterion was later generalized in \cite{DEL} to all possible couplings. Asymptotics for magnetic graphs was studied in \cite{EL3}. The paper \cite{Li1} shows how to find the constant by the leading term of the asymptotics for non-Weyl graphs and gives bounds on this constant. 

In the current note we continue in investigating this problem and illustrate the main results of \cite{Li1} on several simple examples with standard coupling. In particular, we focus on the pseudo orbit expansion and explain to the reader in detail how to construct the resonance condition by this method. The paper is structured as follows. In the second section we introduce the model and the resonance asymptotics. The third section deals with the pseudo orbit expansion. In the section~4 we give two theorems on the effective size of an equilateral graph. In all these three section we give theorems without proofs, which can be found in the referred papers. The only exception is the theorem~\ref{thm-vertexscat}, where a simpler proof than the general one from \cite{Li1} is given. The last three section are devoted the the examples of non-Weyl graphs. In all of them the resonance condition is obtained from a regular directed graph and then from a simplified graph after deleting some of its edges.

\section{Preliminaries}
\begin{figure}
\centering
\includegraphics[width=0.5\textwidth]{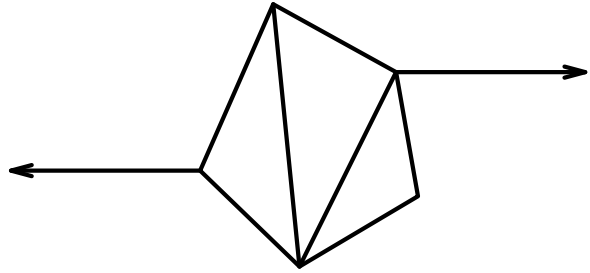}
\caption{Example of a non-compact graph with two halflines, which are denoted by lines with arrows.}
\label{fig0}
\end{figure} 
We assume a quantum graph with attached halflines. We consider a metric graph~$\Gamma$, which consists of the set of vertices~$\mathcal{V}$ and the set of edges~$\mathcal{E}$. There are $N$ finite edges, parametrized by the segments $(0,\ell_j)$, and $M$ infinite edges, parametrized by $(0,\infty)$. If $M \geq 1$, then we call the graph non-compact and by its compact part we denote the union of finite internal edges. Example of this non-compact graph is shown in figure~\ref{fig0}. The graph is equipped with the quantum Hamiltonian $H$ acting as negative second derivative with the domain consisting of functions in the Sobolev space $W^{2,2}(\Gamma)$ (this space consists of Sobolev spaces on the edges) which fulfill the coupling conditions at the vertices
\begin{equation}
  (U_j-I)\Psi_j + i (U_j+ I)\Psi_j' = 0\,.\label{eq-coupl}
\end{equation}
Here $U_j$ is a $d\times d$ unitary matrix ($d$ is the degree of the given vertex), $I$ is $d\times d$ unit matrix, $\Psi_j$ is the vector of limits of functional values at the given vertex and $\Psi_j'$ is the vector of limits of outgoing derivatives.

In this paper, we will mostly consider standard coupling (sometimes also known as Kirchhoff, free or Neumann)~-- in this case the functional values are continuous in the vertex and the sum of outgoing derivatives is equal to zero. The corresponding coupling matrix is $U_j = \frac{2}{d}J-I$, where $J$ is a $d\times d$ matrix with all entries equal to one.

We will study \emph{resolvent resonances}. The resolvent resonance is often defined as pole of the meromorphic continuation of the resolvent $(H-\lambda \mathrm{id})^{-1}$. We will use a simpler definition, proof that both definitions are equivalent can be done by the method of complex scaling \cite{EL1}.
\begin{definition}
We say that $\lambda  = k^2$ is a resolvent resonance of $H$ if there is a generalized eigenfunction $f$, which satisfies $ -f''(x) = k^2 f(x)$ on all edges of the graph, satisfies the coupling conditions (\ref{eq-coupl}) at the vertices and its restriction to each halfline is $\beta_j \,\mathrm{exp}(ikx)$.
\end{definition}

\begin{definition}
The counting function $N(R)$ gives the number of all resolvent resonances (including multiplicities) in the circle of radius~$R$ centered at the origin in the $k$-plane.
\end{definition}

Now we state the surprising result of Davies and Pushnitski \cite{DP} on the existence of graphs with non-Weyl asymptotics.
\begin{theorem}
For graphs with standard coupling the following bound on the counting function holds
$$
  N(R) = \frac{2}{\pi} W R+\mathcal{O}(1)\quad \mathrm{as\ }R\to \infty
$$
with $0\leq W \leq \mathrm{vol\,}\Gamma$, where $\mathrm{vol\,}\Gamma$ is the sum of the lengths of the internal edges of the graph. $W$ is strictly smaller than $\mathrm{vol\,}\Gamma$ iff there exists a balanced vertex, the vertex for which there is the same number of internal and external edges.
\end{theorem}

\section{Pseudo orbit expansion for the resonance condition}
In this section, we lay theoretical grounds for the method of pseudo orbit expansion. The method of pseudo orbit expansion was used to finding the spectrum and trace formula for compact quantum graphs (we refer to \cite{BHJ}). The method was recently adjusted by the author to obtain the resonance condition for graphs with attached halflines \cite{Li1}. We explain here the method, in most cases omitting the proofs, which can be found in the above references.

We define an effective vertex-scattering matrix, which gives effective scattering from the internal edges to other internal edges emanating from a vertex with radiating condition on the halflines.

\begin{definition}
Let us assume a vertex $v$ of the graph with $n$ internal edges, which emanate from this vertex, parametrized by $(0,\ell_j)$ with $x=0$ corresponding to $v$ and $m$ halflines. Let the solution of the Schr\"odinger equation on these internal edges be $f_j(x) = \alpha_j^{\mathrm{in}}\,\mathrm{exp\,}(-ikx)+\alpha_j^{\mathrm{out}}\,\mathrm{exp\,}(ikx)$ and on the external edges $g_s(x) = \beta_s\,\mathrm{exp\,}(ikx)$. Then the effective vertex-scattering matrix $\tilde \sigma^{(v)}$ is given by the relation $\vec \alpha_v^{\mathrm{out}} = \tilde \sigma^{(v)}\vec \alpha_v^{\mathrm{in}}$, where the vectors $\vec \alpha_v^{\mathrm{out}}$ and $\vec \alpha_v^{\mathrm{in}}$ have as entries the above coefficients of the outgoing and incoming waves, respectively.
\end{definition}

\begin{theorem}\label{thm-vertexscat}
The effective vertex-scattering matrix for the vertex $v$ with $n$ internal and $m$ external edges and the standard coupling is $\tilde \sigma^{(v)} = \frac{2}{n+m}J_n-I_n$, where $J_n$ is $n\times n$ matrix with all entries equal to one and $I_n$ is $n\times n$ unit matrix. In particular, for a balanced vertex $\tilde \sigma^{(v)} = \frac{1}{n}J_n-I_n$.
\end{theorem}
\begin{proof}
The theorem is proven as corollary 4.3 in \cite{Li1}, we will show here a direct proof. The coupling condition yield
\begin{eqnarray*}
\alpha_j^{\mathrm{out}} + \alpha_j^{\mathrm{in}}  = \alpha_i^{\mathrm{out}} + \alpha_i^{\mathrm{in}} = \beta_s\quad \forall i, j = 1,\dots n, \quad \forall s = 1,\dots, m\,,\\
ik \sum_{j=1}^n (\alpha_j^{\mathrm{out}} - \alpha_j^{\mathrm{in}} )+ ik \sum_{s=1}^m\beta_j = 0\,.
\end{eqnarray*}
Now we fix $i$ and substitute for $\beta_s = \alpha_i^{\mathrm{out}} + \alpha_i^{\mathrm{in}}$ and $\alpha_j^{\mathrm{out}}  = \alpha_i^{\mathrm{out}} + \alpha_i^{\mathrm{in}} - \alpha_j^{\mathrm{in}}$. We obtain
$$
  \sum_{j=1}^n (\alpha_i^{\mathrm{out}} + \alpha_i^{\mathrm{in}} - 2\alpha_j^{\mathrm{in}})+ m (\alpha_i^{\mathrm{out}} + \alpha_i^{\mathrm{in}}) = 0\,.
$$
From this equation we have
$$
  \alpha_i^{\mathrm{out}} = \frac{2}{n+m}\left(\sum_{j=1}^n \alpha_j^{\mathrm{in}}\right) -\alpha_i^{\mathrm{in}}
$$
from which the result follows.
\end{proof}

Now we introduce the oriented graph $\Gamma_2$, which is made from the compact part of the graph $\Gamma$. Each edge is replaced by two oriented edges $b_j$, $\hat b_j$ of the same length and opposite directions. For illustration see figure~\ref{fig1}. We will define the following matrices.

\begin{definition}
The $2N\times 2N$ matrix $\tilde \Sigma$ is a block-diagonalizable matrix written in the basis corresponding to 
$$
  \vec{\alpha} = (\alpha_{b_1}^\mathrm{in},\dots, \alpha_{b_N}^\mathrm{in},\alpha_{\hat{b}_1}^\mathrm{in},\dots,
\alpha_{\hat{b}_N}^\mathrm{in})^\mathrm{T} 
$$
which is block diagonal with blocks $\tilde \sigma_v$ if transformed to the basis
$$
(\alpha_{b_{v_{1}1}}^\mathrm{in},\dots,\alpha_{b_{v_{1}d_1}}^\mathrm{in},\alpha_{b_{v_{2
}1}}^\mathrm{in},\dots,\alpha_{b_{v_{2}d_2}}^\mathrm{in},\dots)^\mathrm{T}\,,
$$
where $b_{v_{1}j}$ is the $j$-th edge ending at the vertex $v_1$. 

Moreover, we define $2N\times 2N$ matrix $Q = \begin{pmatrix}0& I_N\\I_N&0\end{pmatrix}$, the scattering matrix $S = Q \tilde \Sigma$ and 
$L =\mathrm{diag\,}(\ell_1,\dots , \ell_N,\ell_1,\dots , \ell_N)$.
\end{definition}

Note that $\tilde\Sigma$ and $S$ may for general coupling be energy dependent. However, this is not the case for standard coupling, since the matrix $\tilde \sigma$ is not energy dependent. The matrix $S = Q \tilde \Sigma$ is constructed in the following way. We denote its first $N$ rows by $b_1,\dots b_N$ and the other $N$ rows by the edges in the opposite direction $\hat b_1,\dots \hat b_N$; similarly we denote the columns. If $b_j$ ends in the vertex $v$, then we write into the $b_j$-th column and all rows corresponding to oriented edges starting from $v$ the entries of the vertex-scattering matrix $\tilde \sigma^{(v)}$. To the $\hat b_j$-th row we write the diagonal term of $\tilde \sigma^{(v)}$, to the other rows which correspond to the edges emanating from $v$ the nondiagonal terms; in the rest of the rows in this column is zero.

We continue with the theorem which is proven in \cite{Li1}, but the proof is with the exception of the effective vertex-scattering matrix the same as e.g. in \cite{BHJ}.

\begin{theorem}
The resonance condition is 
$$
  \mathrm{det\,}(\mathrm{e}^{ikL} Q \tilde\Sigma - I_{2N}) = 0\,,
$$ 
where $I_{2N}$ is a $2N\times 2N$ unit matrix.
\end{theorem}

Now we define periodic orbits, pseudo orbits and irreducible pseudo orbits. 

\begin{definition}
A \emph{periodic orbit} $\gamma$ on the graph $\Gamma_2$ is a closed path on $\Gamma_2$. A \emph{pseudo orbit} $\tilde \gamma$ is a collection of periodic orbits. An \emph{irreducible pseudo orbit} $\bar \gamma$ is a pseudo orbit, which does not use any directed edge more than once. We define length of a periodic orbit by $\ell_\gamma = \sum_{j, b_j\in \gamma} \ell_j$; the length of pseudo orbit (and hence irreducible pseudo orbit) is the sum of the lengths of the periodic orbits from which it is composed. We define product of scattering amplitudes for a periodic orbit $\gamma = (b_1, b_2, \dots, b_n)$ (it uses first the bond $b_1$, then it continues to $b_2$, etc., it ends in the bond $b_n$ which is connected to $b_1$) as $A_\gamma = S_{b_2b_1}S_{b_3b_2}\dots S_{b_1b_n}$, where $S_{b_2b_1}$ is the entry of the matrix $S$ in the $b_2$-th row and $b_1$-th column. For a pseudo orbit we define $A_{\tilde\gamma} =  \Pi_{\gamma_n\in\tilde \gamma } A_{\gamma_j}$. Finally, by $m_{\tilde \gamma}$ we denote the number of periodic orbits in the pseudo orbit $\tilde \gamma$. 
\end{definition}

Now we restate the previous theorem using irreducible pseudo orbits; the proof can be found in \cite{BHJ}. 

\begin{theorem}\label{thm-rescon2}
The resonance condition is given by the sum over irreducible pseudo orbits
$$
  \sum_{\bar \gamma} (-1)^{m_{\bar \gamma}} A_{\bar \gamma} \,\mathrm{e}^{ik\ell_{\bar \gamma}} = 0\,.
$$ 
\end{theorem}

Note that in general $A_{\bar \gamma}$ can be energy dependent, but this is not the case for standard coupling. 

\section{Theorems on the effective size of an equilateral graph}
Now we focus on equilateral graphs~-- the graphs which have all the internal edges of the length $\ell$. First, we give a theorem how to find the effective size of an equilateral graph. Then we introduce a method how to reduce the number of oriented edges of the graph $\Gamma_2$ for an equilateral graph with standard coupling and balanced vertices. Both theorems were proven in \cite{Li1}.

\begin{theorem}\label{thm-effsize}
Let us assume an equilateral graph (internal edges of lengths $\ell$). Then the effective size is $\frac{\ell}{2}n_{\mathrm{nonzero}}$, where $n_{\mathrm{nonzero}}$ is the number of nonzero eigenvalues of the matrix $S = Q\tilde \Sigma$. 
\end{theorem}

\begin{theorem}
Let us assume an equilateral graph $\Gamma$ for which no edge starts and ends in one vertex and no two vertices are connected by more than one edge. Let us assume standard coupling and let there be a balanced vertex $v_2$ in which directed edges $b_1, b_2, \dots, b_d$ end. Then the following construction does not change the resonance condition. We delete the directed edge $b_1$ of the graph $\Gamma_2$, which starts in the vertex $v_1$, and replace it by ``ghost edges'' $b_1', b_1'', \dots, b_1^{(d-1)}$, where the ``ghost edge'' $b_1^{(j)}$ starts in the vertex $v_1$ and continues to the directed edge $b_{j+1}$. Contribution of the irreducible pseudo orbit containing ``ghost edge'' $b_1'$ to the resonance condition given by theorem~\ref{thm-rescon2} is the following. The ghost edge does not contribute to the length of the pseudo orbit. The scattering amplitude from the bond $b$, which ends in $v_1$, to the bond $b_2$ is equal to the scattering amplitude from $b$ to $b_1$ taken with the opposite sign. Every ``ghost edge'' can be in the irreducible pseudo orbit used only once. Similarly, one can delete more edges; for each balanced vertex we delete an edge which ends in this vertex.
\end{theorem}

\section{Example 1: two abscissas and two halflines}

\begin{figure}
\centering
\includegraphics[width=0.85\textwidth]{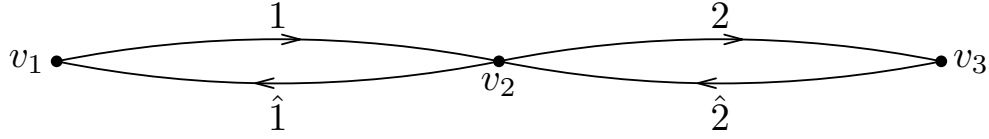}
\caption{Graph $\Gamma_2$ for two abscissas and two halflines. The halflines were ``cut off'' and each internal edge of the graph $\Gamma$ was replaced by two oriented edges.}
\label{fig1}
\end{figure} 

In the following sections we use previous theorems in several simple examples. Graph in the first example consists of two internal edges of length $\ell$ connected in one vertex with two halflines. There is Dirichlet coupling ($f(0) = 0$) at the spare ends (vertices $v_1$ and $v_3$) of the abscissas and standard coupling in the central vertex ($v_2$). The oriented graph $\Gamma_2$ is shown in figure~\ref{fig1}. Since the vertex $v_2$ is balanced, we have by theorem~\ref{thm-vertexscat} $\tilde \sigma^{(v_2)} = \frac{1}{2}\begin{pmatrix}-1&1\\1&-1\end{pmatrix}$. The vertex-scattering matrices for the vertices $v_1$ and $v_3$ are $\tilde \sigma^{(v_1)} = \tilde \sigma^{(v_3)} = -1$. The matrix $S = Q\tilde \Sigma$ is
$$
  \begin{array}{c|cccc}&1&2&\hat 1&\hat 2\\ \hline
       1 & 0 & 0 & -1 & 0\\
       2 & 1/2 & 0 & 0 & -1/2\\
  \hat 1 & -1/2 & 0 & 0 & 1/2\\
  \hat 2 & 0 & -1 & 0 & 0\\
  \end{array}\,.
$$ 
We explicitly mark the edges to which the columns and rows correspond. Eigenvalues of $S$ are $-1$, 1 and 0 with multiplicity 2. From theorem~\ref{thm-effsize} we obtain that the effective size of the graph is $\ell$.

Now we find the resonance condition using pseudo orbits. The contribution of the pseudo orbits which do not include any bond is 1. Clearly, there are no irreducible pseudo orbits on one or three bonds. Let us look at the contribution of the irreducible pseudo orbits on two edges, i.e. find the coefficient by $\mathrm{exp}(2ik\ell)$. We have two irreducible pseudo orbits $(1,\hat 1)$ and $(2, \hat 2)$. The scattering amplitude between $\hat 1$ and $1$ is $-1$, the scattering amplitude between $1$ and $\hat 1$ is $-1/2$. There is one periodic orbit in the pseudo orbit $(1,\hat 1)$, hence there is a factor of $(-1)^{1}$. The contribution of the irreducible pseudo orbit $(1,\hat 1)$ is $(-1)(-1/2)(-1) = -1/2$, similarly for the pseudo orbit $(2, \hat 2)$. Hence the coefficient by $\mathrm{exp}(2ik\ell)$ is $-1$. Finally, we find the contribution the irreducible pseudo orbits on four edges. There are two irreducible pseudo orbits: $(1,2,\hat 2, \hat 1)$ and $(1,\hat 1)(2,\hat 2)$; the former consists of one periodic orbit, the latter of two periodic orbits. The contribution of the irreducible pseudo orbit $(1,2,\hat 2, \hat 1)$ is $(-1)^2(1/2)^2(-1) = -1/4$, the contribution of the irreducible pseudo orbit $(1,\hat 1)(2,\hat 2)$ is $(-1)^2(-1/2)^2(-1)^2 = 1/4$. Hence the coefficient by $\mathrm{exp}(4ik\ell)$ is 0, because contribution of the two irreducible pseudo orbits cancel. The resonance condition is $1-\mathrm{exp}(2ik\ell) = 0$.

\begin{figure}
\centering
\includegraphics[width=0.85\textwidth]{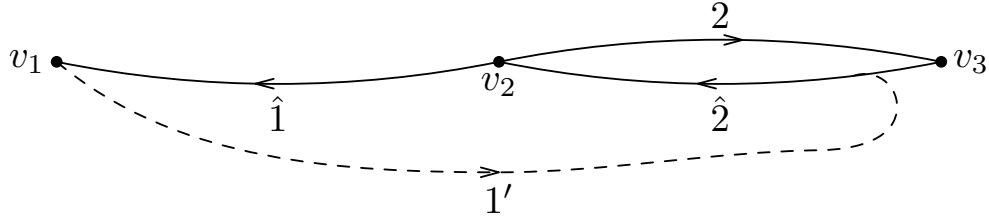}
\caption{Graph $\Gamma_2$ after deleting the edge 1}
\label{fig2}
\end{figure} 

Since the vertex $v_2$ is balanced, we can also delete one bond which ends in this vertex, say bond 1. We replace it by a ``ghost edge'' $1'$ which starts at $v_1$ and continues to the only other edge which ends in $v_2$, the bond $\hat 2$ (see figure~\ref{fig2}). Now we can do pseudo orbit expansion again. The contribution of pseudo orbit which does not include any bond is 1. We have two irreducible pseudo orbits on two ``non-ghost'' edges $(\hat 1,1',\hat 2)$ and $(2, \hat 2)$. The former has contribution $1\cdot 1/2 (-1)^1 = -1/2$ (we have used the fact that the scattering amplitude between $\hat 1$ and $\hat 2$ is $+1$, because the scattering amplitude between $\hat 1$ and 1 was $-1$), the contribution of the latter is $(-1)(-1/2)(-1)^1 = -1/2$. Again we obtain that the coefficient by $\mathrm{exp}(2ik\ell)$ is $-1$. There is no irreducible pseudo orbit, which would use all the three remaining ``non-ghost'' edges, even if it would use the ``ghost edge''. Clearly, there cannot be also any pseudo orbit on four edges, because we have deleted one. Again, we obtain the same resonance condition. In this case it was easier to find that the coefficient by $\mathrm{exp}(4ik\ell)$ is zero.

We conclude that the resolvent resonances in this case are only eigenvalues $\lambda = k^2$ with $k = n \pi$, $n\in \mathbb{Z}$.

\section{Example 2: triangle with attached halflines}

\begin{figure}
\begin{minipage}[b]{0.45\linewidth}
\centering
\includegraphics[width=\textwidth]{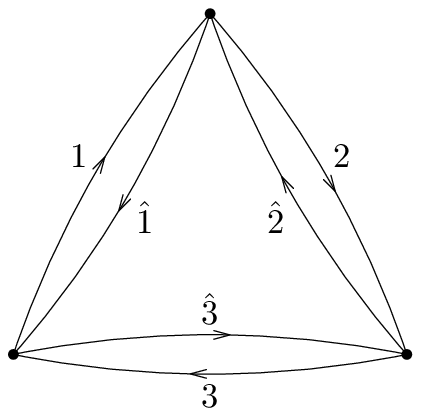}
\caption{Graph $\Gamma_2$ for the triangle with attached halflines}
\label{fig3}
\end{minipage}
\hspace{0.5cm}
\begin{minipage}[b]{0.45\linewidth}
\centering
\includegraphics[width=\textwidth]{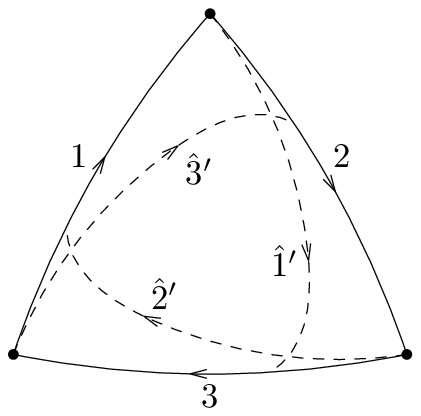}
\caption{Graph $\Gamma_2$ for the triangle with attached halflines after deleting the edges $\hat 1$, $\hat 2$, $\hat 3$}
\label{fig4}
\end{minipage}
\end{figure} 

Let us consider a graph with three internal edges of the lengths $\ell$ in the triangle. To each vertex two halflines are attached, so every vertex is balanced. The graph $\Gamma_2$ is shown in figure~\ref{fig3}. The vertex scattering matrices are in all the vertices $\tilde \sigma^{(v)} = \frac{1}{2}\begin{pmatrix}-1&1\\1&-1\end{pmatrix}$. The matrix $S = Q\tilde \Sigma$ is
$$
  \begin{array}{c|cccccc}&1&2&3&\hat 1&\hat 2&\hat 3\\ \hline
       1 & 0 & 0 & 1/2 & -1/2 & 0 & 0\\
       2 & 1/2 & 0 & 0 & 0 & -1/2 & 0\\
       3 & 0 & 1/2 & 0 & 0 & 0 & -1/2\\
  \hat 1 & -1/2 & 0 & 0 & 0 & 1/2 & 0\\
  \hat 2 & 0 & -1/2 & 0 & 0 & 0 & 1/2\\
  \hat 3 & 0 & 0 & -1/2 & 1/2 & 0 & 0\\
  \end{array}\,.
$$ 
Its eigenvalues are 1, $-1/2$ with multiplicity 2 and 0 with multiplicity 3. Hence the effective size of this graph is $3\ell/2$.

Now we find the contributions of the pseudo orbits to the resonance condition. There is no irreducible pseudo orbit on 1 edge and on 5 edges. We have the following three irreducible pseudo orbits on two edges: $(1, \hat 1)$; $(2, \hat 2)$; $(3, \hat 3)$. There are two irreducible pseudo orbits on three edges $(1,2,3)$ and $(\hat 1, \hat 3, \hat 2)$, six on four edges $(1,\hat 1)(2,\hat 2)$; $(1,\hat 1)(3,\hat 3)$; $(3,\hat 3)(2,\hat 2)$; $(1,2,\hat 2, \hat 1)$; $(2,3,\hat 3, \hat 2)$; $(3,1, \hat 1,\hat 3)$ and eight irreducible pseudo orbits on six edges $(1,\hat 1)(2,\hat 2)(3,\hat 3)$; $(1,2,\hat 2, \hat 1)(3,\hat 3)$; $(2,3,\hat 3, \hat 2)(1,\hat 1)$; $(3,1,\hat 1, \hat 3)(2,\hat 2)$; $(1,2,3)(\hat 1, \hat 3,\hat 2)$; $(1,2,3,\hat 3, \hat 2,\hat 1)$; $(2,3, 1, \hat 1, \hat 3, \hat 2)$; $(3,1,2,\hat 2,\hat 1,\hat 3)$. Below, we compute their contributions
\begin{eqnarray*}
\mathrm{exp\,}0 :& &1\,,\\
\mathrm{exp\,}(2ik\ell) :& & \left(-\frac{1}{2}\right)^2 (-1)^1 \cdot 3 = -\frac{3}{4}\,,\\
\mathrm{exp\,}(3ik\ell) :& &\left(\frac{1}{2}\right)^3 (-1)^1 \cdot 2 = -\frac{1}{4}\,,\\
\mathrm{exp\,}(4ik\ell) :& &\left(-\frac{1}{2}\right)^4 (-1)^2 \cdot 3 + \left(-\frac{1}{2}\right)^2\left(\frac{1}{2}\right)^2 (-1)^1 \cdot 3 = 0\,,\\
\mathrm{exp\,}(6ik\ell) :& &\left(-\frac{1}{2}\right)^6 (-1)^3 + \left(-\frac{1}{2}\right)^2\left(\frac{1}{2}\right)^2 \left(-\frac{1}{2}\right)^2 (-1)^2 \cdot 3 +\\
 & & + \left(\frac{1}{2}\right)^3\left(\frac{1}{2}\right)^3 (-1)^2 + \left(-\frac{1}{2}\right)^2\left(\frac{1}{2}\right)^4  (-1)^1 \cdot 3 = 0\,.\\
\end{eqnarray*}
The resonance condition is 
$$
  1-\frac{3}{4}\,\mathrm{exp\,}(2ik\ell)-\frac{1}{4}\,\mathrm{exp\,}(3ik\ell) = 0\,.
$$

The alternative way of constructing the resonance condition is using the method of deleting the edges. We have three balanced vertices, hence we can delete the edges $\hat 1$, $\hat 2$ and $\hat 3$ and replace them by ``ghost edges'' $\hat 1'$, $\hat 2'$ and $\hat 3'$ (see figure~\ref{fig4}). It is clear that there are no irreducible pseudo orbits on 4, 5 or 6 edges and since no ``ghost edge'' continues to an edge ending in a vertex from which this ``ghost edge'' starts, there are also no pseudo orbits on one edge. There are three irreducible pseudo orbits on two ``non-ghost edges'' $(1,2,\hat 2')$; $(2,3,\hat 3')$; $(3,1,\hat 1')$ and there are two irreducible pseudo orbits on three edges $(1,2,3)$ and $(1,\hat 1',3,\hat 3',2,\hat 2')$. Their contributions are
\begin{eqnarray*}
\mathrm{exp\,}0 :& &1\,,\\
\mathrm{exp\,}(2ik\ell) :& & \left(\frac{1}{2}\right)^2 (-1)^1 \cdot 3 = -\frac{3}{4}\,,\\
\mathrm{exp\,}(3ik\ell) :& &\left(\frac{1}{2}\right)^3 (-1)^1 \cdot 2 = -\frac{1}{4}\,,\\
\end{eqnarray*}
which gives the same resonance condition. The resonances are such $\lambda = k^2$ with $k = 2n\pi/\ell$ and $k = (\pi + 2 n \pi-i \ln{2})/\ell$ with multiplicity two, $n\in \mathbb{Z}$.

\section{Example 3: square with attached halflines}

\begin{figure}
\begin{minipage}[b]{0.45\linewidth}
\centering
\includegraphics[width=\textwidth]{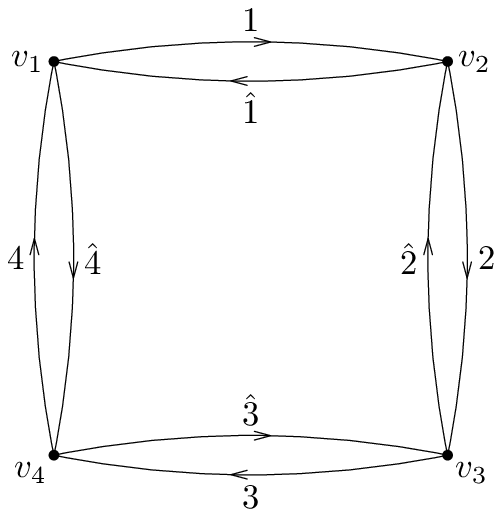}
\caption{Graph $\Gamma_2$ for the square with attached halflines}
\label{fig5}
\end{minipage}
\hspace{0.5cm}
\begin{minipage}[b]{0.45\linewidth}
\centering
\includegraphics[width=\textwidth]{fig6}
\caption{Graph $\Gamma_2$ for the square with attached halflines after deleting the edges $\hat 1$, $\hat 2$, $\hat 3$, $\hat 4$}
\label{fig6}
\end{minipage}
\end{figure} 

We consider a square of the edges of lengths $\ell$, in each vertex two halflines are attached, hence every vertex is balanced. The graph $\Gamma_2$ is in figure~\ref{fig5}. The vertex-scattering matrices are again $\tilde \sigma^{(v)} = \frac{1}{2}\begin{pmatrix}-1&1\\1&-1\end{pmatrix}$. The matrix $S = Q\tilde\Sigma$ is 
$$
  \begin{array}{c|cccccccc}&1&2&3&4&\hat 1&\hat 2&\hat 3&\hat 4\\ \hline
       1 & 0 & 0 & 0 & 1/2 & -1/2 & 0 & 0 & 0\\
       2 & 1/2 & 0 & 0 & 0 & 0 & -1/2 & 0 & 0\\
       3 & 0 & 1/2 & 0 & 0 & 0 & 0 & -1/2 & 0\\
       4 & 0 & 0 & 1/2 & 0 & 0 & 0 & 0 & -1/2\\
  \hat 1 & -1/2 & 0 & 0 & 0 & 0 & 1/2 & 0 & 0\\
  \hat 2 & 0 & -1/2 & 0 & 0 & 0 & 0 & 1/2 & 0\\
  \hat 3 & 0 & 0 & -1/2 & 0 & 0 & 0 & 0 & 1/2\\
  \hat 4 & 0 & 0 & 0 & -1/2 & 1/2 & 0 & 0 & 0\\
  \end{array}\,,
$$ 
its eigenvalues are $1$, $-1$ and 0 with multiplicity 6; the effective size is $\ell$.

One can see that there are no irreducible pseudo orbits on odd number of edges. The irreducible pseudo orbits on two edges are $(1,\hat 1)$; $(2,\hat 2)$; $(3,\hat 3)$; $(4,\hat 4)$, on four edges $(1,\hat 1)(2,\hat 2)$; $(1,\hat 1)(3,\hat 3)$; $(1,\hat 1)(4,\hat 4)$; $(2,\hat 2)(3,\hat 3)$; $(2,\hat 2)(4,\hat 4)$; $(3,\hat 3)(4,\hat 4)$; $(1,2,\hat 2, \hat 1)$; $(2,3,\hat 3, \hat 2)$; $(3,4,\hat 4, \hat 3)$; $(4,1,\hat 1, \hat 4)$; $(1,2,3,4)$; $(\hat 4,\hat 3,\hat 2,\hat 1)$, on six edges $(1,\hat 1)(2,\hat 2)(3,\hat 3)$; $(1,\hat 1)(2,\hat 2)(4,\hat 4)$; $(1,\hat 1)(3,\hat 3)(4,\hat 4)$; $(2,\hat 2)(3,\hat 3)(4,\hat 4)$; $(1,2,\hat 2,\hat 1)(3,\hat 3)$; $(1,2,\hat 2,\hat 1)(4,\hat 4)$;\\
 $(2,3,\hat 3,\hat 2)(1,\hat 1)$;  $(2,3,\hat 3,\hat 2)(4,\hat 4)$; $(3,4,\hat 4,\hat 3)(1,\hat 1)$; $(3,4,\hat 4,\hat 3)(2,\hat 2)$; $(4,1,\hat 1,\hat 4)(2,\hat 2)$;\\
 $(4,1,\hat 1,\hat 4)(3,\hat 3)$; $(1,2,3,\hat 3,\hat 2, \hat 1)$; $(2,3,4,\hat 4,\hat 3,\hat 2)$; $(3,4,1,\hat 1,\hat 4, \hat 3)$; $(4,1,2,\hat 2, \hat 1,\hat 4)$ and on eight edges $(1,\hat 1)(2,\hat 2)(3,\hat 3)(4,\hat 4)$; $(1,2,\hat 2,\hat 1)(3,\hat 3)(4,\hat 4)$; $(2,3,\hat 3,\hat 2)(1,\hat 1)(4,\hat 4)$;\\
 $(3,4,\hat 4,\hat 3)(1,\hat 1)(2,\hat 2)$; $(4,1,\hat 1,\hat 4)(2,\hat 2)(3,\hat 3)$; $(1,2,\hat 2,\hat 1)(3,4,\hat 4,\hat 3)$; $(2,3,\hat 3,\hat 1)(4,1,\hat 1,\hat 4)$;\\
 $(1,2,3,\hat 3, \hat 2,\hat 1)(4,\hat 4)$; $(2,3,4,\hat 4, \hat 3,\hat 2)(1,\hat 1)$; $(3,4,1,\hat 1, \hat 4,\hat 3)(2,\hat 2)$; $(4,1,2,\hat 2, \hat 1,\hat 4)(3,\hat 3)$;\\
 $(1,2,3,4)(\hat 4,\hat 3,\hat 2,\hat 1)$; $(1,2,3,4,\hat 4,\hat 3,\hat 2,\hat 1)$; $(2,3,4,1,\hat 1,\hat 4,\hat 3,\hat 2)$; $(3,4,1,2,\hat 2,\hat 1,\hat 4,\hat 3)$;\\
 $(4,1,2,3,\hat 3,\hat 2,\hat 1,\hat 4)$. Their contributions to the resonance conditions are
\begin{eqnarray*}
\mathrm{exp\,}0 :& &1\,,\\
\mathrm{exp\,}(2ik\ell) :& & \left(-\frac{1}{2}\right)^2(-1)^1\cdot 4 = -1\,,\\
\mathrm{exp\,}(4ik\ell) :& & \left(-\frac{1}{2}\right)^2\left(-\frac{1}{2}\right)^2(-1)^2\cdot 6+ \left(-\frac{1}{2}\right)^2\left(\frac{1}{2}\right)^2(-1)^1\cdot 4+ \left(\frac{1}{2}\right)^4(-1)^1\cdot 2 = 0\,\\
\mathrm{exp\,}(6ik\ell) :& & \left(-\frac{1}{2}\right)^2\left(-\frac{1}{2}\right)^2\left(-\frac{1}{2}\right)^2(-1)^3\cdot 4+\left(-\frac{1}{2}\right)^2\left(\frac{1}{2}\right)^2\left(-\frac{1}{2}\right)^2(-1)^2\cdot 8+\\
& & + \left(-\frac{1}{2}\right)^2\left(\frac{1}{2}\right)^4(-1)^1 \cdot 4 = 0\,,\\
\mathrm{exp\,}(8ik\ell) :& & \left(-\frac{1}{2}\right)^2\left(-\frac{1}{2}\right)^2\left(-\frac{1}{2}\right)^2\left(-\frac{1}{2}\right)^2(-1)^4+\left(-\frac{1}{2}\right)^2\left(\frac{1}{2}\right)^2\left(-\frac{1}{2}\right)^2\left(-\frac{1}{2}\right)^2\cdot\\
& & \cdot(-1)^3\cdot 4+\left(-\frac{1}{2}\right)^2\left(\frac{1}{2}\right)^2\left(-\frac{1}{2}\right)^2\left(\frac{1}{2}\right)^2(-1)^2\cdot 2 + \left(-\frac{1}{2}\right)^2\left(\frac{1}{2}\right)^4\cdot\\
& & \cdot\left(-\frac{1}{2}\right)^2(-1)^2\cdot 4+ \left(\frac{1}{2}\right)^4\left(-\frac{1}{2}\right)^4(-1)^2+\left(-\frac{1}{2}\right)^2\left(\frac{1}{2}\right)^6(-1)^1\cdot 4 =0\,.\\
\end{eqnarray*}

\begin{figure}
\centering
\includegraphics[width=0.55\textwidth]{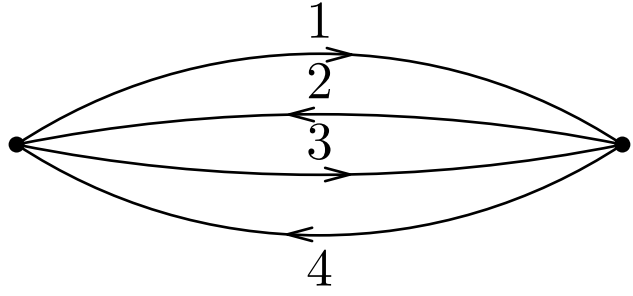}
\caption{Graph $\Gamma_2$ simplified after deleting the edges}
\label{fig7}
\end{figure} 

If we delete the edges $\hat 1$, $\hat 2$, $\hat 3$ and $\hat 4$, we obtain figure~\ref{fig6}. Note that the scattering amplitude for path from 1 to 2 is the same as the scattering amplitude from 1 to 4 through $\hat 1'$; in both cases we obtain $1/2$. Similarly for other vertices, so the graph is equivalent to a graph in figure~\ref{fig7}, where all the scattering amplitudes are $1/2$. This simplifies finding the resonance condition. We have four irreducible pseudo orbits on two edges $(12)$; $(14)$; $(32)$; $(34)$ and four on four edges $(12)(34)$; $(14)(32)$; $(1234)$; $(1432)$. Their contribution is 
\begin{eqnarray*}
\mathrm{exp\,}0 :& &1\,,\\
\mathrm{exp\,}(2ik\ell) :& & \left(\frac{1}{2}\right)^2(-1)^1\cdot 4 = -1\,,\\
\mathrm{exp\,}(4ik\ell) :& & \left(\frac{1}{2}\right)^2\left(\frac{1}{2}\right)^2(-1)^2\cdot 2+ \left(\frac{1}{2}\right)^4(-1)^1\cdot 2 = 0\,.\\
\end{eqnarray*}

The resonance condition is $1-\mathrm{exp\,}(2ikl) = 0$, the positions of resonances are $\lambda = k^2$ with $k = n\pi/\ell$, $n\in \mathbb{Z}$.

\section*{Acknowledgements}
Support of the grant 15-14180Y of the Grant Agency of the Czech Republic is acknowledged. The author thanks to R.~Band for a useful discussion and the referee for useful comments.


\begin{thebibliography}{BHJ12} 

\bibitem[BHJ12]{BHJ}
R.~Band, J.\,M.~Harrison, C.\,H.~Joyner, Finite pseudo orbit expansions for spectral quantities of quantum graphs, {\em J. Phys. A: Math. Theor.\/} {\bf  45}, 325204 (2012). DOI: 10.1088/1751-8113/45/32/325204
  
\bibitem[DEL10]{DEL}
E.\,B.~Davies, P.~Exner, J.~Lipovsk\'{y}, Non-{W}eyl asymptotics for quantum graphs with general coupling conditions, {\em J. Phys. A: Math. Theor.\/} {\bf 43}, 474013 (2010). DOI: 10.1088/1751-8113/43/47/474013
  
\bibitem[DP11]{DP}
E.\,B.~Davies, A.~Pushnitski, Non-{W}eyl resonance asymptotics for quantum graphs, {\em Analysis and PDE\/} {\bf 4}, 729--756 (2011). DOI: 10.2140/apde.2011.4.729

\bibitem[EL07]{EL1}
P.~Exner, J.~Lipovsk\'{y}, Equivalence of resolvent and scattering resonances on quantum graphs, in {\em Adventures in Mathematical Physics (Proceedings, Cergy-Pontoise 2006)\/} vol {\bf 447} (Providence, R.I.), pp 73--81 (2007). DOI: 10.1090/conm/447

\bibitem[EL11]{EL3}
P.~Exner, J.~Lipovsk\'{y}, Non-{W}eyl resonance asymptotics for quantum graphs in a magnetic field, {\em Phys. Lett. A\/} {\bf 375}, 805--807 (2011). DOI: 10.1016/j.physleta.2010.12.042

\bibitem[KS03]{KoS3}
T.~Kottos, U.~Smilansky, Quantum graphs: a simple model for chaotic scattering, {\em J. Phys. A: Math. Gen.} {\bf 36}, 3501-3524, (2003). DOI: 10.1088/0305-4470/36/12/337

\bibitem[KS04]{KSc}
T.~Kottos, H.~Schanz, Statistical properties of resonance width for open quantum systems {\em Waves in Random Media\/} {\bf 14}, S91--S105 (2004). DOI: 10.1088/0959-7174/14/1/013

\bibitem[Lip15]{Li1}
J.~Lipovsk\'{y}, On the effective size of a non-Weyl graph, arXiv:1507.04176 [math-ph] (2015).

\end{thebibliography}
\end{document}